\documentclass[conference]{IEEEtran}
\IEEEoverridecommandlockouts
\usepackage{amsthm}
\usepackage{graphicx}
\usepackage{textcomp}
\usepackage{xcolor}
\usepackage{amssymb}
\usepackage{amsfonts}
\usepackage{fdsymbol}
\newtheorem{definition}{Definition}
 
\newtheorem{theorem}{Theorem}
\newtheorem{lemma}{Lemma}

\usepackage[numbers,sort&compress]{natbib}
\usepackage{bm}
\usepackage{subfigure} 
\makeatletter
\newif\if@restonecol
\makeatother

\usepackage[linesnumbered,ruled,vlined]{algorithm2e}%[ruled,vlined]{
\usepackage{algpseudocode}
\usepackage{amsmath}
\usepackage[ruled,linesnumbered]{algorithm2e}
\usepackage{setspace}

\def\BibTeX{{\rm B\kern-.05em{\sc i\kern-.025em b}\kern-.08em
    T\kern-.1667em\lower.7ex\hbox{E}\kern-.125emX}}
\begin{document}

\title{Differentially Private Data Publication with Multi-level Data Utility}

\author{Honglu~Jiang,~\IEEEmembership{Member,~IEEE,}
        S M Sarwar,
        Haotian Yu,
        and Sheikh Ariful Islam,~\IEEEmembership{Senior Member,~IEEE,}
      
\IEEEcompsocitemizethanks{\IEEEcompsocthanksitem H. Jiang is with the Department of Informatics and Engineering Systems, The University of Texas Rio Grande Valley, Brownsville, TX, USA. Email: honglu.jiang@utrgv.edu.
\IEEEcompsocthanksitem S M Sarwar is with the Department of Computer science, The University of Texas Rio Grande Valley, Edinburg, TX, USA. E-mail: sm.sarwar01@utrgv.edu.
\IEEEcompsocthanksitem H. Yu is with the Department of Data Analytics, The George Washington University, Washington, DC 20052 USA. E-mail: yuxx6789@gwu.edu.
\IEEEcompsocthanksitem S. A. Islam is with the Department of Informatics and Engineering Systems, The University of Texas Rio Grande Valley, Brownsville, TX, USA. Email: sheikhariful.islam@utrgv.edu.}}

%\IEEEoverridecommandlockouts
%\IEEEpubid{\makebox[\columnwidth]{978-1-5386-1027-5/17/\$31.00~\copyright2017 IEEE \hfill} \hspace{\columnsep}\makebox[\columnwidth]{ }}

\maketitle

\begin{abstract}
Conventional private data publication mechanisms aim to retain as much data utility as possible while ensuring sufficient privacy protection on sensitive data. Such data publication schemes implicitly assume that all data analysts and users have the same data access privilege levels. However, it is not applicable for the scenario that data users often have different levels of access to the same data, or different requirements of data utility. The multi-level privacy requirements for different authorization levels pose new challenges for private data publication. Traditional PPDP mechanisms only publish one perturbed and private data copy satisfying some privacy guarantee to provide relatively accurate analysis results. To find a good tradeoff between privacy preservation level and data utility itself is a hard problem, let alone achieving multi-level data utility on this basis. In this paper, we address this challenge in proposing a novel framework of data publication with compressive sensing supporting multi-level utility-privacy tradeoffs, which provides differential privacy. Specifically, we resort to compressive sensing (CS) method to project a $n$-dimensional vector representation of users' data to a lower $m$-dimensional space, and then add deliberately designed noise to satisfy differential privacy. Then, we selectively obfuscate the measurement vector under compressive sensing by adding linearly encoded noise, and provide different data reconstruction algorithms for users with different authorization levels. Extensive experimental results demonstrate that ML-DPCS yields multi-level of data utility for specific users at different authorization levels.
\end{abstract}

\begin{IEEEkeywords}
 Differential privacy, compressive sensing, multi-level utility, privacy preservation.
\end{IEEEkeywords}

\section{Introduction}

 A large amount of data has been generated with the rapid development of information technology, which are often published to third parties for data analysis, recommendations, targeted advertising, and reliable predictions. Publishing data containing personal sensitive information results in an increasing concern on privacy violations. Privacy-preserving data publishing (PPDP) has gained significant attentions in recent years as a promising approach for information sharing while preserving data privacy \cite{FBW}. 

 Generally speaking, commonly used approaches for PPDP \cite{cai2018,cai2019trading,JPJ,WXY,cai2021gen} can be characterized into three categories: encryption technology \cite{GRM,zheng,CTS,JGW,WXY,TTC}, $k$-anonymity \cite{SLA} and its derivative approaches ($l$-diversity \cite{MKV}, $t$-closeness \cite{LLV}), and differential privacy \cite{DCD,cai2019}. Differential privacy has gradually become the \emph{de facto} standard privacy definition and provides a strong privacy guarantee, which rests on a sound mathematical foundation with a formal definition and rigorous proof while making almost no assumption on the attacker’s background knowledge.

Conventional private data publication mechanisms aim to retain as much data utility as possible while ensuring sufficient privacy protection on sensitive data. Such data publication schemes implicitly assume that all data analysts and users (for data analysis, recommendations or reliable predictions) have the same data access privilege. However, it is not applicable for the scenario that data users often have different levels of access to the same data, or different requirements of data utility. For example, data analysts aim to mine data correlations, laws of data change and data trend estimation not the individual-specific information; while banks often need individuals' reliable and truthful information to achieve identity verification. Motivated by this application scenario, it would be desirable to provide a mechanism supporting multi-level de-identification, where the more trusted a data miner is, the high level of data recovery quality.
 
The multi-level privacy requirements for different authorization levels pose new challenges for private data publication. Traditional PPDP mechanisms only publish one perturbed and private data copy satisfying some privacy guarantee to provide relatively accurate analysis results. To find a good tradeoff between privacy preservation level and data utility itself is a hard problem, let alone achieving multi-level data utility on this basis.

Compressive sensing (CS) theory was first proposed by Candes \cite{CEJ} \emph{et al.} and Donoho \cite{DDL}, which has attracted considerable attention. It states that a signal can be represented by far fewer samples compared to conventional data acquisition systems. This is a sensing strategy that enables significantly lower data rate and computation cost in the sensing part. While CS is a mature signal processing technique, it can also be employed as a data hiding technology to provide confidentiality. Moreover, various corresponding decoding algorithms that map an under-sampled set of CS-encoded measurements into a recovery of original signal/data \cite{CVMM} have been proposed. Motivated by this, we employ CS to achieve the data hiding and data reconstruction at the same time.

In this paper, we address the challenge of multi-level utility in proposing a novel framework of data publication with compressive sensing supporting multi-level utility-privacy tradeoffs, which provides differential privacy. Our contributions can be summarized as follows:

\begin{itemize}

\item First, we resort to compressive sensing (CS) method to project a $n$-dimensional vector representation of users' data to a lower $m$-dimensional space, and then add deliberately designed noise to satisfy differential privacy. Then, we expand the scope of differentially private data publication to support multi-level data utility by adding linearly encoded noise to selectively obfuscate the measurement vector after compressive sensing for different level of privacy preservation.

\item We provide different data reconstruction algorithms for users with different authorization levels. A data user with low-level authorization cannot recover the data without the CS measurement matrix, and he can only obtain the differentially private data. A semi-honest user knowing the measurement matrix but without the perturbation (encoding) matrix can only partially recover the data with limited data utility, while the fully authorized users processing both the measurement matrix and perturbation matrix can fully recover the data with high recovery performance.

\item We evaluate the performance of data utility on four real-world datasets in two aspects, the $L_2$ error and the classification error rate of SVM classification on the multi-level recovered data, respectively. Comprehensive experiments on four real-world datasets demonstrate that ML-DPCS achieves multi-level data utility, of which two levels of data utility outperform existing methods while providing differential privacy.
\end{itemize}

The rest of this paper is organized as follows. We provide a literature review in Section \ref{sec:rel}. Section \ref{sec:model} formulates our problem and presents necessary background knowledge on compressive sensing and differential privacy. In Section \ref{sec:our}, we propose our ML-DPCS mechanism for data publication satisfying differential privacy and the detailed multi-level data reconstruction approach based on multi-class compressive sampling. Comprehensive experimental studies on four real-world datasets are presented in Section \ref{sec:eva}. Finally, Section \ref{sec:con} concludes the paper.

\section{Related Work}\label{sec:rel}

Privacy Preserving Data Publishing has gained significant attentions, for which the approaches can be divided into three categories: $k$-anonymity \cite{SLA} and its variations including $l$-diversity \cite{MKV}and $t$-closeness \cite{LLV}, cryptographically secure computation \cite{GRM} and differential privacy \cite{DCD}. There had been many work focusing on publishing differentially private aggregates of sensitive information including various differentially private mechanisms for high-dimensional data publication \cite{QYL,ZCG,CXZ,XRZ,CTT,JPJ,CTS,JGW,WXY,TTC}. A common used approach for high-dimensional data publication is dimensionality reduction based on Bayesian network model \cite{ZCG}, sampling technique \cite{CXZ} and random projection \cite{XRZ,CTT}.

However, all these existing schemes mentioned above do not support multi-level access or data utility of the same published dataset. Xiao \emph{et al.} proposed an algorithm with multilevel uniform perturbation of PPDM and evaluate its performance under the $\rho_1-\rho_2$ privacy measurement \cite{XTC}. In \cite{LCL}, Li \emph{et al.} addressed the challenge of Multilevel Trust PPDM services and proposed an additive perturbation approach through adding random Gaussian noise. Their approach allows a data owner to generate distinctly perturbed copies of the data according to different trust levels. Palanisamy \emph{et al.} designed and implemented multi-level utility-controlled data anonymization scheme based on user access privilege rights in the context of large association graphs \cite{PLL}. In \cite{he2017}, He \emph{et al.} investigated how to optimize the tradeoff between latent-data privacy and customized data utility. They identified an optimization problem seeking a data sanitization strategy to realize the maximum latent-data privacy with customized data utility. In \cite{YAP}, Yamac \emph{et al.} proposed a privacy-preserving method supporting de-identification at multiple privacy levels based on compressive sensing. Specifically, different level of recovery quality (data utility) are provided for users at different authorization levels.

The problem of data hiding under compressive sensing theory has been addressed in \cite{ZXZ,DCH,VGT,PCA}. In \cite{ZXZ}, Zhang \emph{et al.} studied the content reconstruction problem and proposed a watermarking mechanism with flexible self-recovery quality. %The embedded watermark data for content recovery are derived from the original discrete cosine transform (DCT) coefficients of host image and do not contain any additional redundancy. 
In \cite{DCH}, Delpha \emph{et al.} presented a compressive sensing based watermarking solution to ensure the sparse data be secretly embedded using Costa's quantization based hiding scheme. Yamac \emph{et al.} \cite{YCD} provided a data hiding scheme to linearly embed and hide data in compressively sensed signals through an encoding matrix, which can be nonlinearly reconstructed. Recently in \cite{YMM}, Yamac \emph{et al.} proposed a privacy-preserving approach, which supports multi-level de-identification and performs data acquisition, encryption and data hiding based on compressive sensing.

The following aspects distinguish our work from the existing approaches. First, our work focuses on differentially private high-dimensional data publishing, which resorts to compressive sensing to achieve dimension reduction, while existing approaches adopt Bayesian network model or random projection. Specifically, the appealing advantage of compressive sensing lies in that it projects high-dimensional signal into a low-dimensional space and achieves data reconstruction by solving the optimization problem, which
fulfills the purpose of both dimension reduction and two-level data utility. Second, we add Laplace noise to compressive sensing result to satisfy differential privacy and add impulsive noise to partially perturb the published data to achieve multi-level data utility. The measurement matrix and the encoding matrix can be considered as secret keys that users with different authorization level have different access to the keys.

\section{Preliminaries}\label{sec:model}

\subsection{Notation}

In this paper, we use $\mathbb{R}^n$ to denote a coordinate space over the real numbers. Lowercase bold letters denote vectors and uppercase letters represent matrice. Let $\|\bm{x}\|_p$ denote the $l_p$-norm of a vector $\bm{x} \in \mathbb{R}^n$, which can be defined as:

\begin{equation}
\|\bm{x}\|_p=(\sum\limits_{i=1}^n|\bm{x}_i|^{p})^{1/p}
\end{equation}
where $p$ is a positive number. Intuitively, $l_1$-norm gives the Manhattan distance of two points while $l_2$-norm computes the Euclidean distance.

Strictly speaking, $l_0$-norm is not actually a norm, which is a cardinality function to obtain the total number of non-zero elements in a vector. It is formally defined as follows:

\begin{equation}
\|\bm{x}\|_0=\#(i|\bm{x}_i\neq 0)
\end{equation}

\subsection{Compressive Sensing}

Compressive sensing (CS), also known as compressive sampling, is a novel sensing/sampling technique in data acquisition. It efficiently achieves signal sampling at a frequency far lower than the Nyquist sampling theorem requires, and can reconstruct the signal from the compressed samples with high probability. Specifically, CS first linearly projects a sparse or compressible high-dimensional signal into a low-dimensional space to achieve signal sampling and data compression simultaneously, followed by the reconstruction process that recovers the signal from far fewer samples by solving the optimization problem. To make the reconstruction process possible, CS relies on two principles \cite{CMW}: \emph{sparsity}, which requires signals to be sparse or compressible; and \emph{incoherence}, which is applied through the restricted isometric property (RIP) \cite{CTT}.

Assume the original signal $\bm{x}\in \mathbb{R}^n$ is an $n$-dimensional vector, it can be represented by an $n \times n$ \emph{orthonormal} basis $\Psi$. Let $\bm{s}$ be the projection coefficient vector of $\bm{x}$ under the basis $\Psi$, then we have $\bm{x}=\bm{s}\Psi$. We call $\bm{s}$ $S$-sparse if it has at most $S$ nonzero entries, that is, $\|s\|_0\leq S$. Therefore, a vector $\bm{x}$ has an $S$-sparse representation if there is an orthonormal basis $\Psi$, where $\bm{x}$'s representation $\bm{s}$ is $S$-sparse (compressible) under $\Psi$. While sparsity of the signal is one of the prerequisites for the application of compressive sensing, fortunately, most natural signals exhibit sparsity or compressibility, and non-strict sparse signals can be approximately regarded as sparse.

For an original signal $\bm{x}$ with a sparse representation, it can be linearly sampled using a CS measurement matrix $\Phi \in \mathbb{R}^{m\times n}$ ($m\ll n$) to obtain an $m$-dimension result: 
\begin{equation}
\bm{y}=\Phi\bm{x}=\Phi \Psi \bm{s}
\end{equation}
of which $A=\Phi \Psi$ is the sensing matrix. The measurement matrix $\Phi$ is employed for signal sampling, while the sensing matrix $A$ is used for signal reconstruction.

$S$-Restricted isometry Property (RIP) implies that it is possible to main the data's statistical characteristics while projecting the original data into a low-dimensional space. To introduce this notion, $\Phi_T, T\subset \{1,\cdots,m\}$ is the $n\times |T|$ sub-matrix obtained by extracting the columns of $\Phi$ corresponding to the indices in $T$. $\delta_S$ is the $S$-restricted isometry constant of $\Phi$, which is the smallest quantity such that:

\begin{equation}
(1-\delta_S)\|\bm{s}\|_{l_2}^2\leq \|\Phi_T \bm{s}\|_{l_2}^2\leq (1+\delta_S)\|\bm{s}\|_{l_2}^2
\end{equation}

for all $S$-sparse vectors $\bm{s}$.

Since $m\ll n$, the sampling process can be considered as the process of projecting a sparse or compressible high-dimensional signal into a low-dimensional space, that is, the process of data compression. Then, we can recover the approximate value $\tilde{\bm{s}}$ of the original data signal $\bm{x}$ by solving the $l_0$-optimization problem:

\begin{equation}
\begin{aligned}
&\tilde{\bm{s}}=\arg \min_{\bm{s}} \|\bm{s}\|_0 \\
&\mathrm{s.t.} \quad \bm{y}=A\bm{s}
\end{aligned}
\end{equation}
of which $\|\bm{s}\|_0 $ denotes the $l_0$-norm of $\bm{s}$. We can obtain the recovered signal $\tilde{\bm{x}}$ based on $\tilde{\bm{x}}=\Psi \tilde{\bm{s}}$.

$l_0$-optimization tries to minimize the $l_0$-norm of a vector corresponding to some constraints, which has been a topic of growing interest in many applications including compressive sensing. However, such an optimization problem is non-convex and known to be NP-hard \cite{NBK}. In many cases, $l_0$-optimization problem could be relaxed to be higher-order norm problem such as $l_1$-optimization or $l_2$-optimization:

\begin{equation}
\begin{aligned}
&\tilde{\bm{s}}=\arg \min_{\bm{s}} \|\bm{s}\|_1 \\
&\mathrm{s.t.} \quad \bm{y}=A\bm{s}
\end{aligned}
\end{equation}

\subsection{Differential Privacy}

\emph{Differential privacy} (DP) was proposed by Dwork \cite{DCD} to provide a strong privacy guarantee and protect against the privacy disclosure of statistical databases.
 It has become the \emph{de facto} standard of privacy preservation, which ensures query results of a dataset are insensitive to the change of a single record. That is, whether a single record exists in the dataset has little effect on the output distribution of the analytical results. Differential privacy is defined based on the neighboring datasets $D$ and $D'$, where $D'$ differs from $D$ by only one record:

\begin{definition}[Differential privacy \cite{DMN}] \label{def:DP}
A randomized algorithm $M$ is $\epsilon$-differentially private if for any pair of neighboring datasets $D$ and $D'$, and for all sets $S$ of possible outputs, we have
\begin{equation}
Pr[M(D)\in S]\leq e^{\epsilon}Pr[M(D')\in S],\nonumber
\end{equation}
where $\epsilon$ is often a small positive real number. 
\end{definition}

In this definition, $\epsilon$ is called \emph{privacy budget}. The smaller the $\epsilon$, the higher the level of privacy preservation. A smaller $\epsilon$ provides greater privacy preservation at the cost of lower data accuracy with more additional noise.

Differential privacy can be achieved by two best known mechanisms, namely the \emph{Laplace mechanism} \cite{DMN} and \emph{exponential mechanism} \cite{MTK}, which are respectively proposed for numerical and query results. We provide the formal definition of Laplace mechanism as follows:

\begin{definition}[Laplace mechanism \cite{DMN}] \label{def:lap}

Given a random algorithm $A$ with the input dataset $D$ and function $f: D\rightarrow R^{d}$ with global sensitivity $GS_{f}$, $A(D)=f(D)+Z$ is $\epsilon$-differentially private, where $Z\sim Lap(GS_{f}/\epsilon)$.
\end{definition}

\section{Proposed Mechanism}\label{sec:our}

In this paper, we employ compressive sensing to achieve data compression and dimensionality reduction, through projecting the data into a lower-dimensional space and 
simultaneously reducing the data correlation of the original data. Then the resulting data is perturbed by adding noise to provide differential privacy, then adding impulsive noise to embed perturbed message to the compressed data. Combining compressive sampling technology with differential privacy allows us to use less noise than previous differentially private mechanisms. 

The proposed mechanism has an advantage over existing privacy-preserving methods, namely, it supports multiple level of data utility for different end users. Specifically, for a data user with low authorization level without knowing the CS measurement matrix $A$ can only obtain the perturbed data $y_w$ to conduct data analysis; for semi-honest third-party users, they access the differentially private data and can partially recover the data with high data utility guarantee; while the authorized users can de-identify and fully recover the data using the keys of both the measurement matrix $A$ and encoding matrix $B$.

This section provides a detailed elaboration of ML-DPCS for high-dimensional data publication satisfying \textbf{d}ifferential \textbf{p}rivacy based on \textbf{c}ompressive \textbf{s}ampling with \textbf{m}ulti-\textbf{l}evel data utility. The framework of ML-DPCS is shown in Figure 1.

\begin{figure}[!htb]
 \centering
  \includegraphics[width=0.48\textwidth]{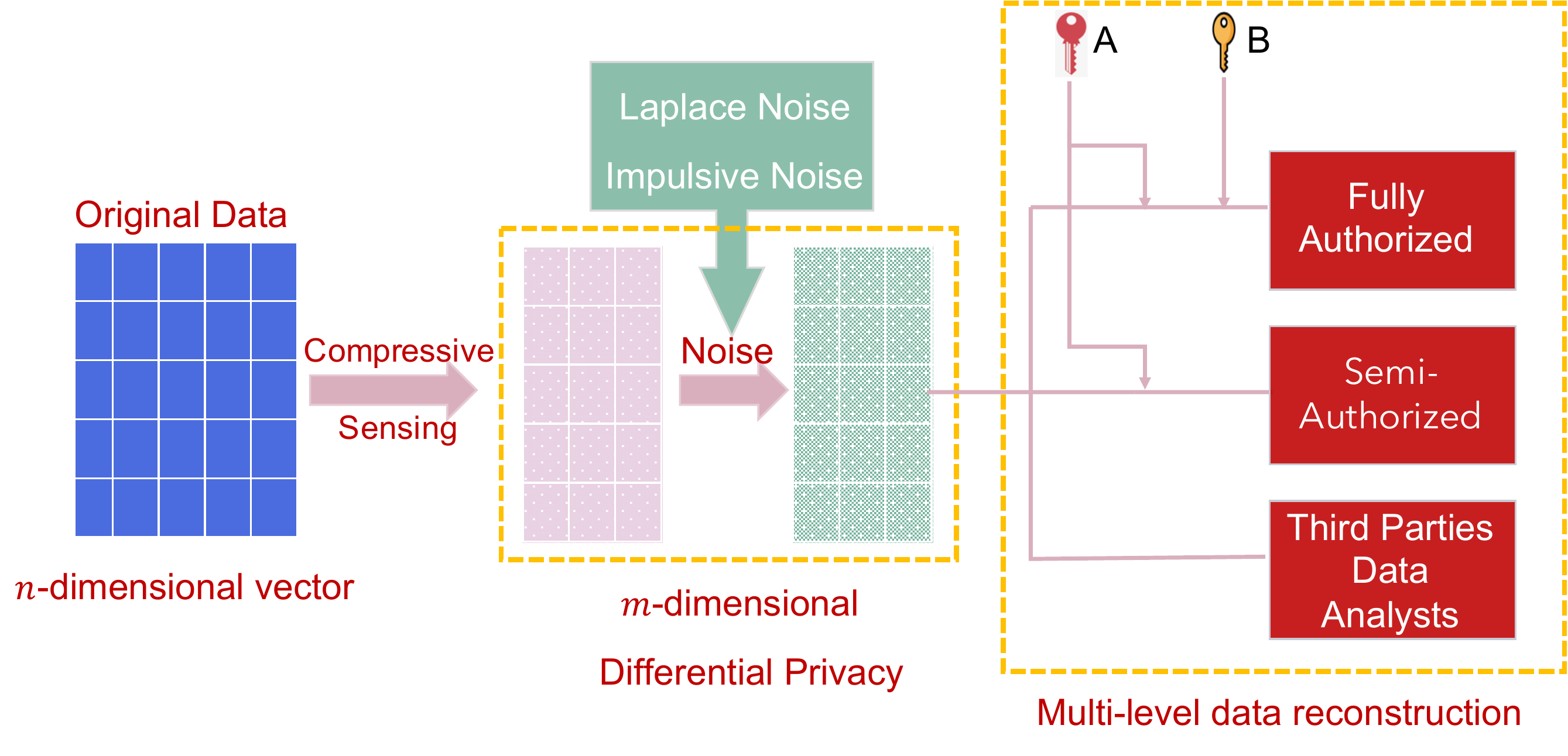}
  \caption{Overview of ML-DPCS.}
  \label{fig:overview}
\end{figure}

\subsection{Data Perturbation based on Compressive Mechanism}\label{bay}

The key idea of ML-DPCS is to project a $n\times n$ dataset into a lower-dimensional $n\times m$ dataset through applying the sampling operator $\Phi$ to obtain 
the sample result $\bm{y}=\Phi\bm{x}$, $y \in \mathbb{R}^m$, of which $\bm{x}$ has a sparse representation and can be linearly sampled as $\bm{x}=\Psi\bm{s}$. Then the result $\bm{y}$ is perturbed by adding noise $\bm{z}$ to guarantee the differential privacy. That is $\tilde{\bm{y}}=\bm{y}+\bm{z}$. If the noise is bounded by $\|\bm{z}\|_2\leq \delta$, it is possible to reconstruct $s$ by solving the optimization problem:

\begin{equation}\label{opt}
\begin{aligned}
&\tilde{\bm{s}}=\arg \min_{\bm{s}} \|\bm{s}\|_1 \\
&\mathrm{s.t.} \quad \|\bm{y}-\bm{As}\|_2\leq \delta
\end{aligned}
\end{equation}

This optimization problem in equation \ref{opt} is known as Basis Pursuit Denoising (BPDN) \cite{DET}.

In order to achieve differential privacy, we add the noise drawn randomly and independently from Laplace distribution, depending on the $L_2$ sensitivity of the measurement matrix $\Phi$. The definition of $L_2$ sensitivity of $\Phi$ is presented in Equation \eqref{sen1} and \eqref{sen2}. 
\begin{equation}\label{sen1}
\begin{aligned}
\Delta q=\max_{1\leq i \leq n} \sqrt{\sum_{j=1}^{m}|\Phi|^2}.
\end{aligned}
\end{equation}

Equivalently, $\Delta q$ can also be defined as:
 
 \begin{equation}\label{sen2}
 \Delta q=\max_{e_i}\|e_i \Phi\|_2
\end{equation}
where $\{e_i\}_{i=1}^{n}$ are standard basis unit vectors. Correspondingly, we add random noise to each entry of $\bm{y}$ and each entry in $\bm{z}$ is drawn randomly and independently from the Laplace distribution $Lap(\sqrt{m})/\epsilon$. That is, for each $\bm{y}[i]$, let $\tilde{\bm{y}[i]}=\bm{y}[i]+ \bm{z}[i]$, where $\bm{z}[i]\sim Lap(\sqrt{m})/\epsilon$.

The detailed procedure is presented in Algorithm \ref{ag1}.

\begin{algorithm}
 \begin{small}
\caption{Differentially Private Compressive Mechanism}
\label{ag1}

\LinesNumbered
\KwIn{Original data $\bm{x}$ with a sparse basis $\Psi$, privacy budget $\epsilon$}
\KwOut{$\tilde{\bm{y}}$}
 
Construct a normalized matrix $\Phi \in \mathbb{R}^m$ with i.i.d. Gaussian distribution;

Compute $\bm{y}=\Phi\bm{x}$;

Construct a $n\times m$ noise matrix $\bm{z}$ with i.i.d. $Lap(\sqrt{m})/\epsilon$,

Compute differentially private $\tilde{\bm{y}}=\bm{y}+\bm{z}$;

\textbf{return} $\tilde{\bm{y}}$
 \end{small}
\end{algorithm}

Moreover, we consider measurement vector corrupted by both Laplace and impulsive noise. The compressively sensed measurements could be linearly encoded by directly adding the encoded messages onto the measurement vector $\bm{y}$. The encrypted data can be denoted as:
\begin{equation}
y_w=A\bm{s}+B\bm{w}+\bm{z}
\end{equation}
of which $\bm{w}\in \{-a,a\}^{M}$ is an $M$ bit dense message with $M<m<n$, and $B$ is an $m\times M$ coding matrix generated by a secret seed shared by the encoder and decoder. There is an embedding power construction constraint $\|B\bm{w}\|_{l_2}\leq P$ to remain the quantization range of the data, then $a$ can be chosen accordingly.

\begin{theorem}
The compressive mechanism ML-DPCS satisfies $\epsilon$-differential privacy.
\end{theorem}

\begin{proof}

The sampling process can be characterized as a random projection from $n$-dimensional data to $m$-dimensional space, $g=\bm{y}=\Phi\bm{x}$. The sensitivity of $g$ is $\Delta_g=\sqrt m$. We add noise to each entry of $\bm{y}$ making use of the Laplacian mechanism, which is $\epsilon$-differentially private according to Definition \ref{def:lap}.

\end{proof}

\subsection{Multi-level Data Reconstruction}\label{bay}

In this subsection, we provide the data reconstruction process in detail for data users with different authorization levels.

%While ML-DPCS projects a dataset into a lower-dimensional space, the utility guarantee should meet these requirements, namely: (1) the vectors which are close in the original space should be likely to remain close in the projection space; (2) similarly, the vectors which are far apart should be likely far apart after compressive sensing. What we should do is to compute a measurement of utility of the $\tilde{\bm{y}}$. Let $\|i-j\|_2$ denote the Euclidean distance between two vectors $i$ and $j$ in the original dimensional space, and $(\|i-j\|_2)^2$ is the squared $L_2$ distance. We adopt the definition of the squared distance between the two vectors in the $m$-dimensional space in \cite{XRZ}, $U^{2}(\tilde{i},\tilde{j})=(\|\tilde{i}-\tilde{j}\|_2)^2-2m\sigma^2$, which has been proved that $E(U^{2}(\tilde{i},\tilde{j}))=(\|i-j\|_2)^2$.

 On the receiver end, a data user not knowing the CS measurement matrix $A$ can only obtain the perturbed data $\bm{y_w}$ to conduct data analysis.

\subsubsection{Semi-honest users}
  For a semi-authorized user, who only has the key $A$ can apply the $l_1$-reconstruction scheme to recover the data as shown in Algorithm \ref{ag2}. 

 Based on the optimization problem in Equation \ref{opt}, if the RIP is satisfied with $\delta_{2S}<\sqrt{2}-1$, it is possible to approximate $x$ with a bounded error:

 \begin{equation}
\|\bm{x}-\bm{x^*}\|_{l_2}\leq C_0\delta
\end{equation}
where $C_0$ depends on on $\delta_{2S}$.

\begin{algorithm}
 \begin{small}
\caption{Reconstruction for a semi-authorized user}
\setstretch{1.35}
\label{ag2}

\LinesNumbered
\KwIn{$\bm{y}$, $\Phi$, $\Psi$, A\\
  Hyper-parameter: $\delta$}
\KwOut{$\bm{x^*}$}

Reconstruct $\bm{s^*}$: $\bm{s^*}=\arg \min\limits_{{\bm{s}}} \|\bm{s}\|_1$, $\mathrm{s.t.} \quad \|\bm{y_w}-\bm{As}\|_2\leq \delta$

Compute $\bm{x^*}=\Psi\bm{s^*} $

\textbf{return} $\bm{x^*}$
\end{small}
\end{algorithm}

 While the compressive mechanism satisfies differential privacy, we need to provide the data utility analysis of $\bm{x^*}$ to demonstrate $\bm{x^*}$ is close to $\bm{x}$, that is, providing utility guarantee. We adopt the squared Euclidean distance between two vectors to measure the utility guarantee after the privacy transformation.

\begin{lemma}\cite{LZW} 
The compressive mechanism in Algorithm \ref{ag1} provides data utility $\|\bm{x^*}-\bm{x}\|_2=O(log(n)/\epsilon)$ with a high probability.
\end{lemma}

\subsubsection{Fully authorized users}

For a fully authorized user possessing both the CS measurement matrix $A$ and encoding matrix $B$ can fully recover the data, we follow the recovery method proposed in \cite{YSG}. First, a matrix $F$ is constructed to annihilates the matrix $B$ such that $FB=0$, of which $F\in R^{T\times m}$ and $T=m-M$. Then, apply $F$ to $y_w=A\bm{s}+B\bm{w}+\bm{z}$:

\begin{equation}
\bm{y'}=F(A\bm{s}+B\bm{w}+\bm{z})=FA\bm{s}+F\bm{z}
\end{equation}

Let $\tilde{\bm{z}}=F\bm{z}$, a $k$-sparse signal $\tilde{\bm{s}}$ can be estimated via

\begin{equation}
\begin{aligned}
&\tilde{\bm{s}}=\arg \min_{\bm{s}} \|\bm{s}\|_1 \\        
&\mathrm{s.t.} \quad \|\bm{y'}-FA\bm{s}\|_2\leq \delta
\end{aligned}
\end{equation}

Then, based on the $\tilde{\bm{s}}$ computed in last step, the estimation of $\bm{w}$ can be computed by using of least-squares method:

\begin{equation}
\begin{aligned}
\bm{w'}=(B^\mathrm{T}B)^{-1}B^\mathrm{T}(\bm{y_w}-A\tilde{\bm{s}})
\end{aligned}
\end{equation}

Since $w_i$ is either $-a$ or $a$, an improved estimate of $\bm{w}$ is $\bm{w''}_i=a*sgn(\bm{w'}_i)$.

Therefore, an improved estimate of $\bm{s}$ can be recovered as:

\begin{equation}
\begin{aligned}
&\bm{s^*}=\arg \min_{\bm{s}} \|\bm{s}\|_1 \\
&\mathrm{s.t.} \quad \|(\bm{y_w}-B\bm{w''})-A\bm{s}\|_2\leq \delta
\end{aligned}
\end{equation}

The original data can be constructed as $\bm{x^*}=\Psi\bm{s^*}$. 
The recovery algorithm for the fully authorized user is provided in Algorithm \ref{ag3}

\begin{algorithm}
\begin{small}
\caption{Reconstruction for a fully authorized user \cite{YCD}}
\setstretch{1.35}
\label{ag3}
\LinesNumbered
\KwIn{$\bm{y_w}$, $\Phi$, $\Psi$, A, B\\
  Hyper-parameter: $\delta$}
\KwOut{$\bm{x^*}$}

Apply $F$ to $\bm{y_w}$: $\bm{y'}=F(A\bm{s}+B\bm{w}+\bm{z})=FA\bm{s}+F\bm{z}$

Solve the estimation $\tilde{\bm{s}}$: $\tilde{\bm{s}}=\arg \min_{\bm{s}} \|\bm{s}\|_1 \quad\mathrm{s.t.} \quad \|\bm{y'}-FA\bm{s}\|_2\leq \delta$

Compute $\bm{w'}$: $\bm{w'}=(B^\mathrm{T}B)^{-1}B^\mathrm{T}(\bm{y_w}-A\tilde{\bm{s}})$

Compute the threshold $\bm{w''}$: $\bm{w''}_i=a*sgn(\bm{w'}_i)$

Compute an improved estimation $\bm{s^*}$: $\bm{s^*}=\arg \min_{\bm{s}} \|\bm{s}\|_1 \quad \mathrm{s.t.} \quad \|(\bm{y_w}-B\bm{w''})-A\bm{s}\|_2\leq \delta$

Compute $\bm{x^*}=\Psi\bm{s^*}$

\textbf{return} $\bm{x^*}$
\end{small}
\end{algorithm}

\section{Experimental Evaluations}\label{sec:eva}

In this section, we conduct extensive experiments to demonstrate the performance of our ML-DPCS mechanism and compare it with two benchmark approaches, PrivBayes \cite{ZCPC} and DPPro \cite{XRZ}, on four real-world datasets of NLTCS \cite{nltcs}, ACS \cite{acs}, BR2000 \cite{br2000} and Adult \cite{adult}. The data utility is evaluated by the errors of the original data and the published data with different utility levels, and the classification error rate of the SVM classification on the perturbed datasets.

\subsection{Experimental Settings}\label{sec:datasets}

\subsubsection{Datasets} We make use of four real-world datasets in our experiments: NLTCS \cite{nltcs} consists records of $21574$ individuals participated in the National Long Term Care Survey, and each record has 16 attributes; ACS \cite{acs} includes $47461$ records of personal information from the $2013$ and $2014$ ACS sample sets in IPUMS-USA, where each record has $23$ attributes; BR2000 \cite{br2000} consists of $38000$ census records with $14$ attributes collected from Brazil in the year 2000; and Adult \cite{adult} contains personal information such as gender, salary, and education level of $45222$ records extracted from the 1994 US Census, where each record has 15 attributes. The first two datasets only contain binary attribute values while the last two possess continuous as well as categorical attributes with multiple values. We summarize the statistics of these datasets in Table \ref{tab:data}. 
\begin{table}[!htb]
\renewcommand{\arraystretch}{1.4}
\caption{Data Characteristics} 
\centering
\begin{tabular}{|c|c|c|c|}
\hline
\textbf{Dataset}& \textbf{Cardinality}   &    \textbf{Dimensionality}         &   \textbf{Domain size} \\    
\hline 
\hline                    
NLTCS &  $21574$ &  $16$ & $\approx 2^{16}$   \\
\hline
ACS &  $ 47461$ &  $23$ & $\approx 2^{23}$ \\
\hline
BR2000 &  $ 38000$ &  $14$  & $\approx 2^{32} $ \\
\hline
Adult &  $ 45222$ &  $15$ & $\approx 2^{52}$ \\                      
\hline
\end{tabular}
\label{tab:data}
\end{table}

\subsubsection{Evaluation Metrics} We consider two tasks to evaluate the performance of ML-DPCS. The first task is to study the errors of the output data and the original data, namely $\|\bm{x}-\bm{x^*}\|_2$. Here, we only compute the $L_2$ error between original data and the partially recovered data of semi-honest users and the fully recovered data of fully authorized users, and only compare this with the result of PrivBayes since the published perturbed data of the ML-DPCS and the output of DPPro are in  lower dimensional space. A lower $L_2$ error represents high data utility.

The second task is to evaluate the classification results of SVM classifiers. The purpose of data publication is to conduct data analysis and data mining. We adopt SVM to evaluate the data utility from the perspective of data applications, as SVM is the most popular classification approach among various data mining techniques with powerful discriminative features both in linear and non-linear classifications \cite{SNSR}.  Specifically, we train two classifiers on ACS to predict whether an individual: (1) goes to school, (2) lives in a multi-generation family; four classifiers are constructed on NLTCS to predict whether an individual: (1) is unable to manage money, (2) is unable to bathe; two classifiers are trained on BR2000 to predict whether an individual (1) owns a private dwelling, (2) is a Catholic; and two classifiers are trained on Adult to predict whether an individual (1) is a female, (2) makes over $50$K a year. For each classifier, we use $80\%$ of the tuples of the dataset for training and the other $20\%$ as the testing set. The prediction accuracy of each SVM classifier is measured by the \emph{misclassification rate} on the testing set.

\subsubsection{Parameter Values}

Recall that $S$ is the data sparsity, $n$ is the dimensionality of the original data, $m$ is the compressed dimensionality and $M$ is the length of the embedded data.
Based on the compressive sensing theory and determined after many trials, we set the data sparsity $S=m/6$, $m/n=0.5$ and $M/m=0.2$, of which $m/n$ is the measurement rate and $M/m$ represents the embedding rate. A larger measurement rate (smaller compression ratios) and smaller embedding rate can obtain better performance.

\subsection{Experimental Results}\label{sec:results} 

We carry out $50$ independent runs for each of the experiments mentioned above. In this subsection, we report the averaged results of these experiments for statistical confidence.

\subsubsection{Results on $L_2$ Error}

For the task of examining the accuracy of $L_2$ Error, we report the $E\{\|\bm{x}-\bm{x^*}\|_2\}$ values on the partially recovered data (P-Our) and fully of recovered data (F-Our) of our ML-DPCS, and compare with that of PrivBayes under a varying privacy budget $\epsilon$ from $0.01$ to $1.6$.

\begin{figure}[!htb]
\begin{minipage}[t]{0.23\textwidth}
\centering
\subfigure[ACS, $L_2$ Error]{
\includegraphics[width=\textwidth]{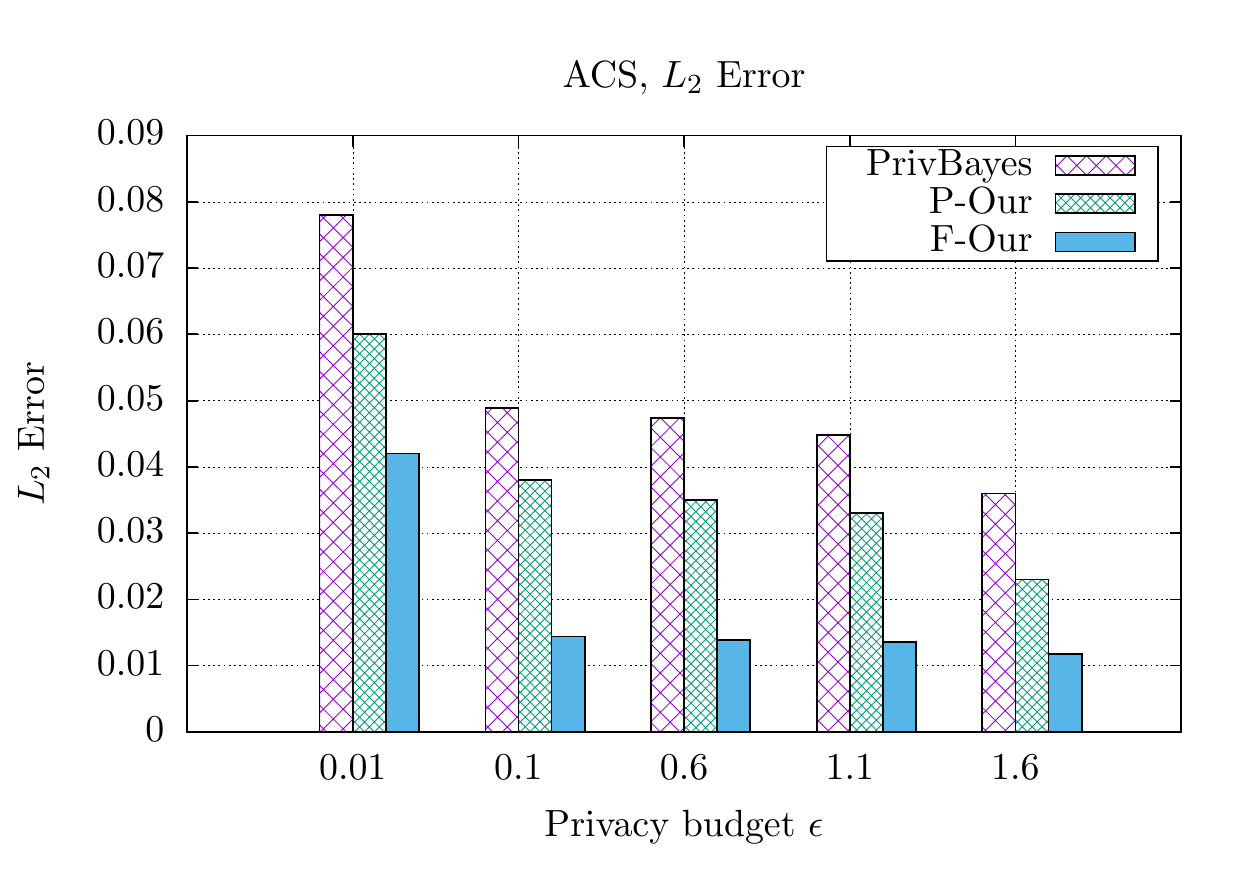}}
\end{minipage}
\begin{minipage}[t]{0.24\textwidth}
\centering
 \subfigure[NLTCS, $L_2$ Error]{
\includegraphics[width=\textwidth]{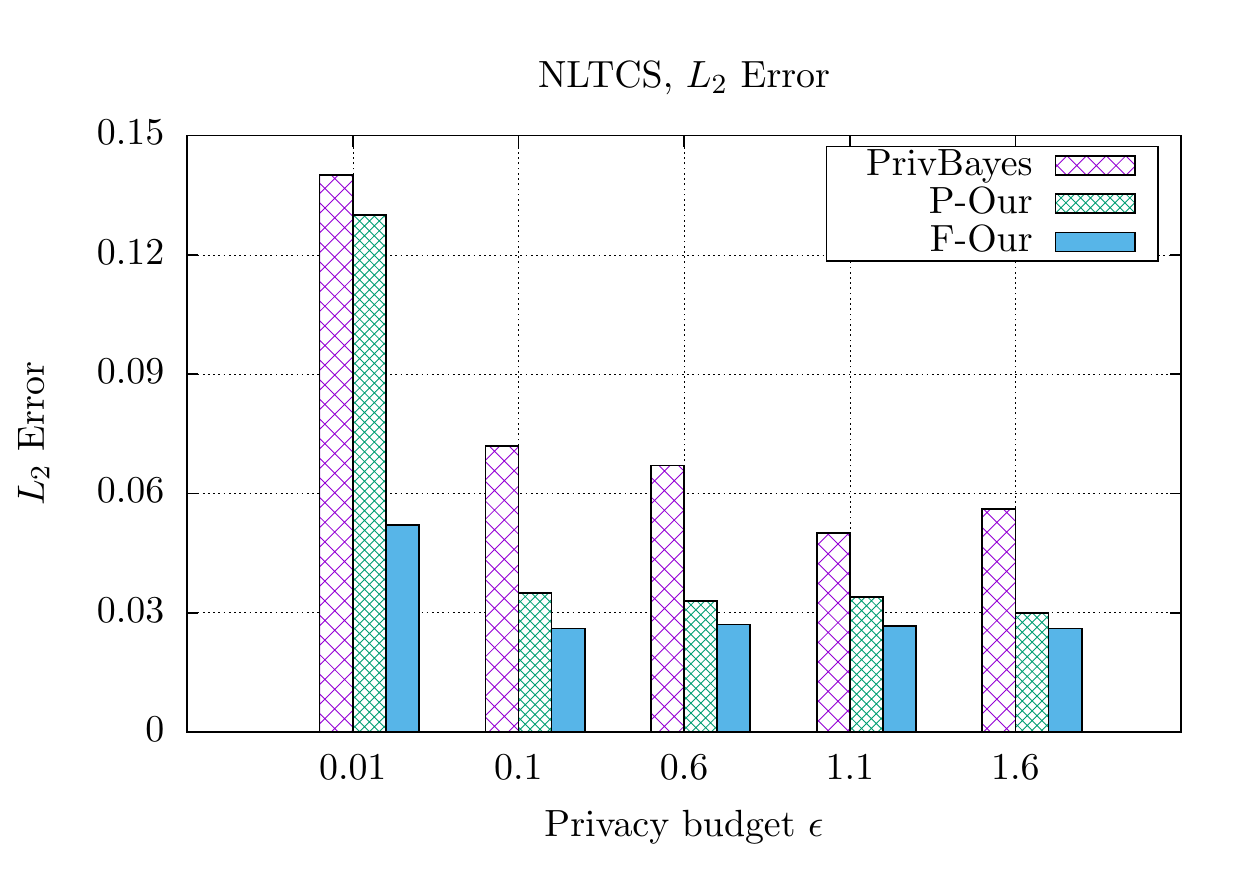}}
\end{minipage}

\begin{minipage}[t]{0.24\textwidth}
\centering
 \subfigure[BR2000, $L_2$ Error]{
\includegraphics[width=\textwidth]{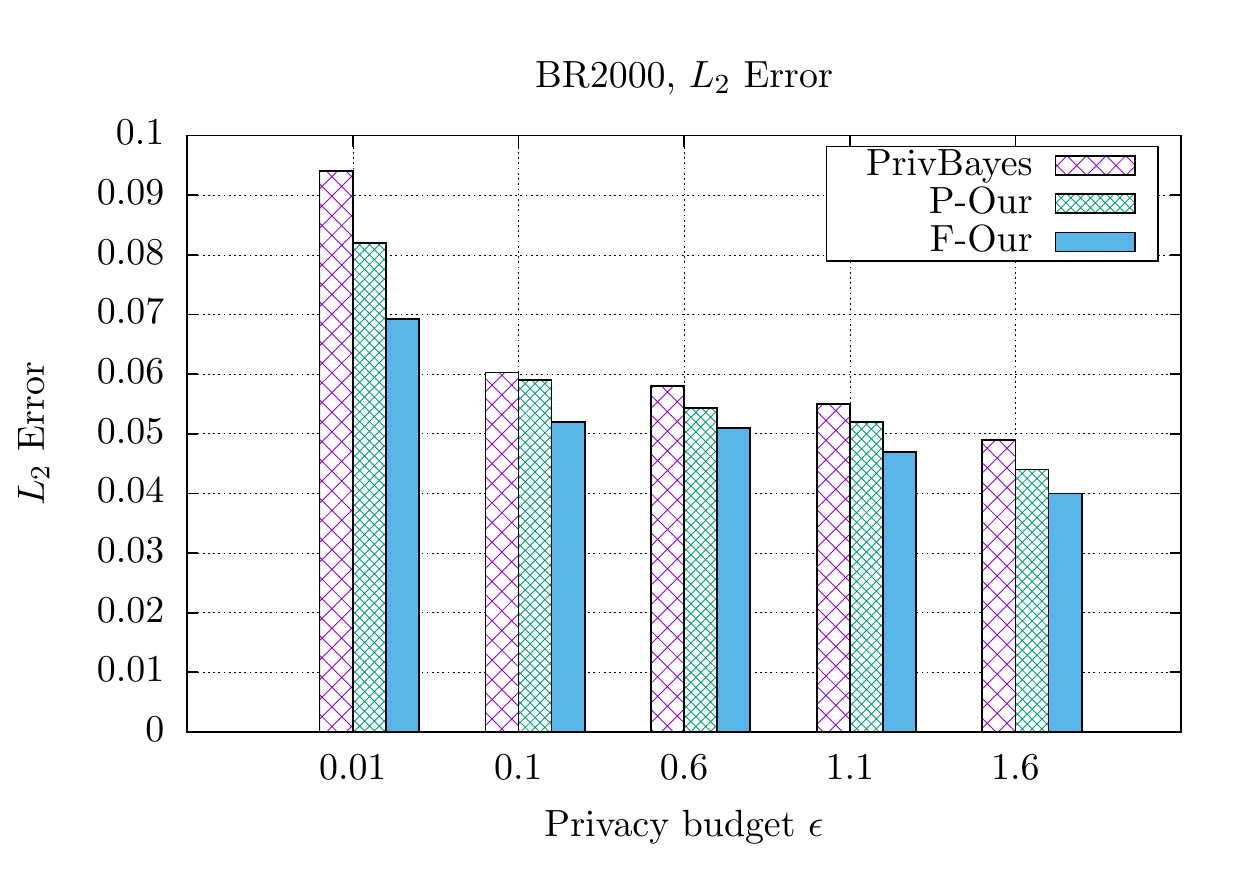}}
\end{minipage}
\begin{minipage}[t]{0.24\textwidth}
\centering
\subfigure[Adult, $L_2$ Error]{
\includegraphics[width=\textwidth]{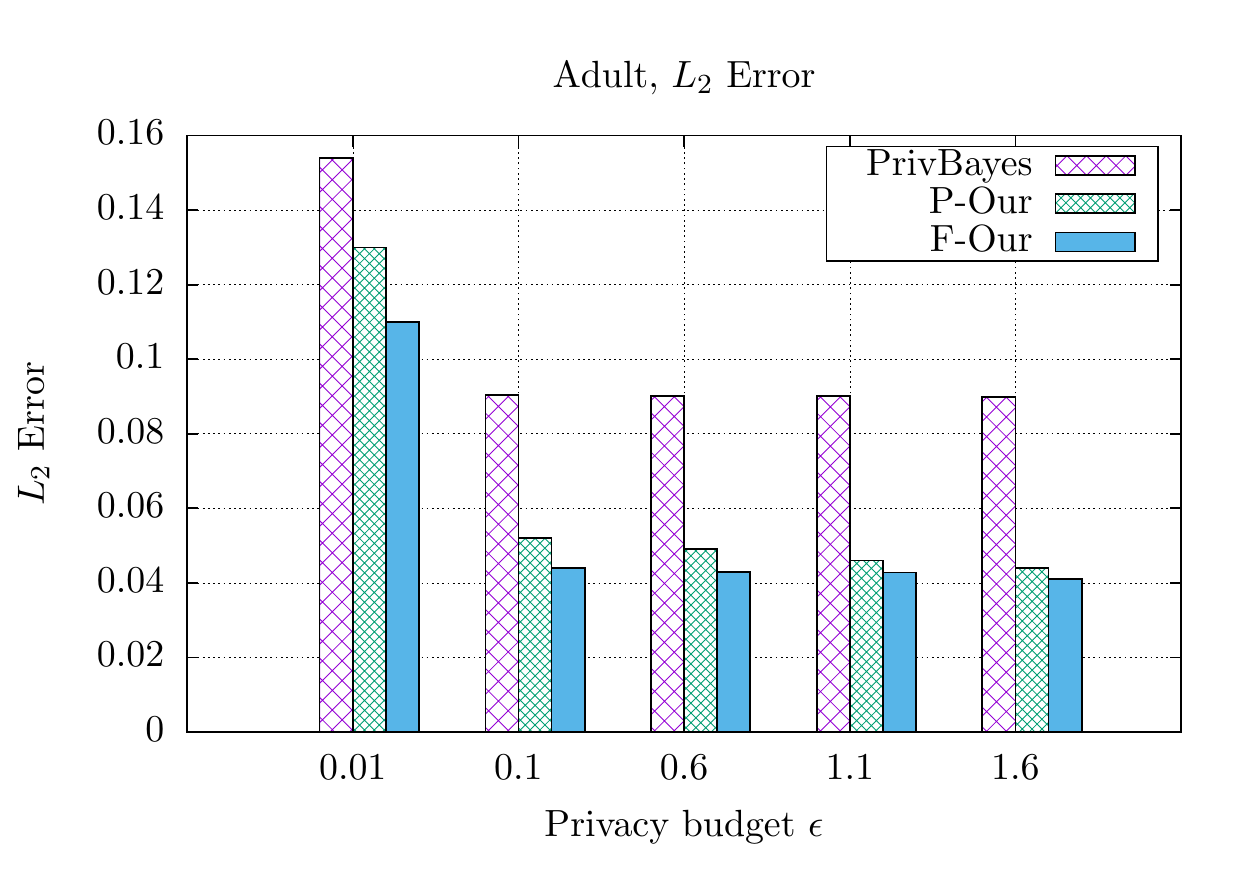}}
\end{minipage}

\caption{Results of $L_2$ error with different $\epsilon$.}
\label{fig4}
\end{figure}

Figure \ref{fig4} shows the average $L_2$ errors of different data on the four datasets, respectively. From Figure \ref{fig4}, one can see that the average $L_2$ error of these approaches decrease when $\epsilon$ increases over the four datasets. It is obvious that when $\epsilon$ is larger, smaller noise is required, and the data utility is higher. We can observe that not only the partially recovered data but also the fully recovered data outperform PrivBayes in all cases, while the relative superiority of the fully recovered data is more pronounced when $\epsilon$ is small. 

\subsubsection{Results on SVM classification}

\begin{figure}[!htp]
\begin{minipage}[t]{0.24\textwidth}
\centering
\subfigure[ACS, Y=School]{
\includegraphics[width=\textwidth]{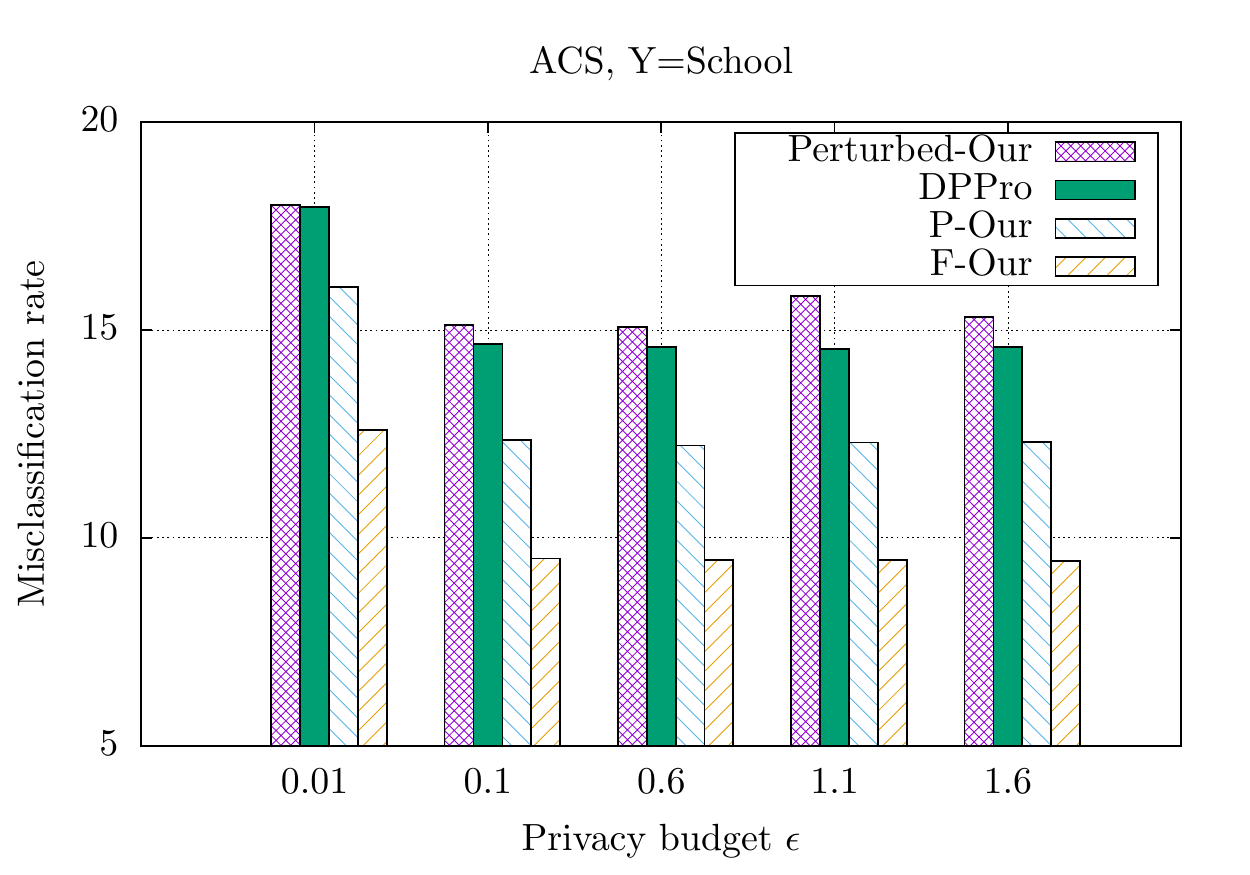}}
\end{minipage}
\begin{minipage}[t]{0.24\textwidth}
\centering
\subfigure[ACS, Y=Multi-gen]{
\includegraphics[width=\textwidth]{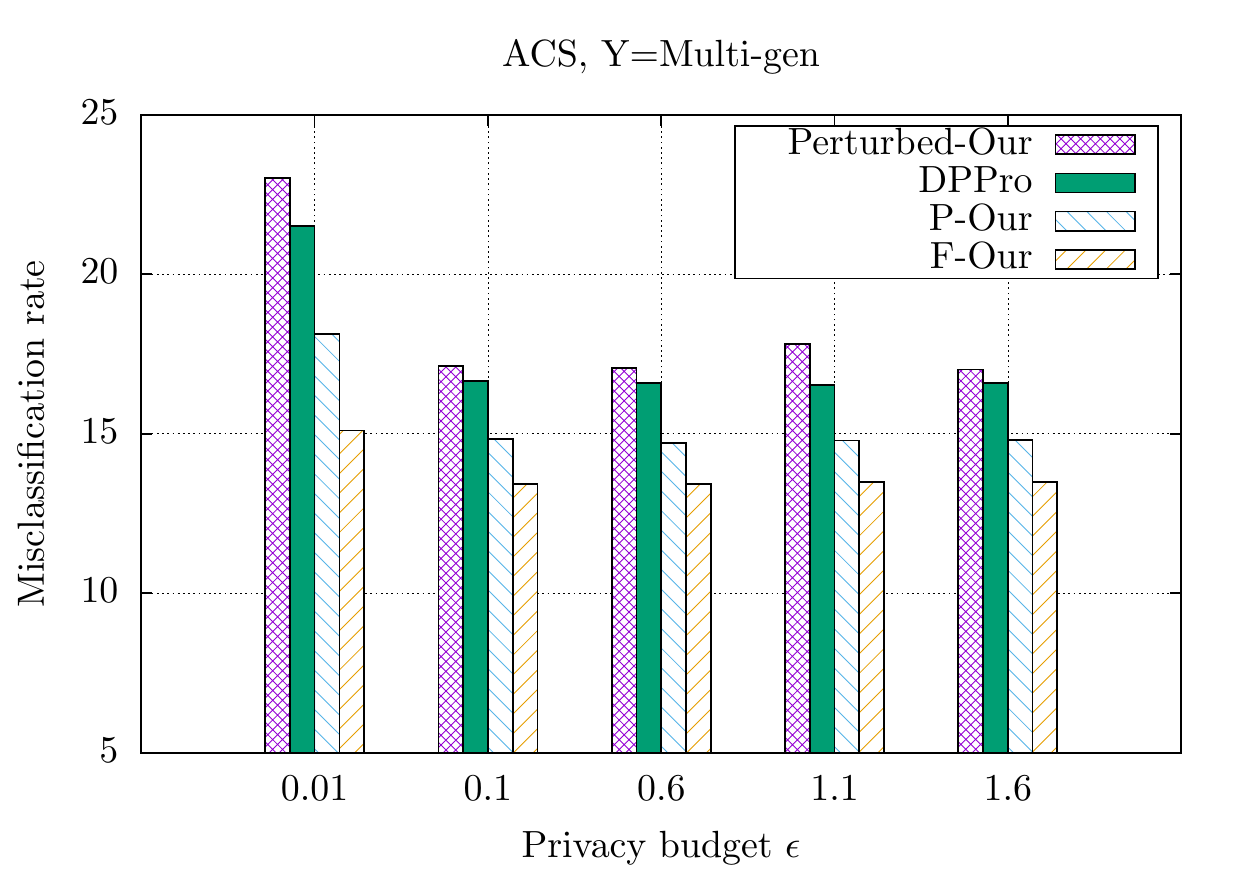}}
\end{minipage}

\begin{minipage}[t]{0.24\textwidth}
\centering
\subfigure[NLTCS, Y=Money]{
\includegraphics[width=\textwidth]{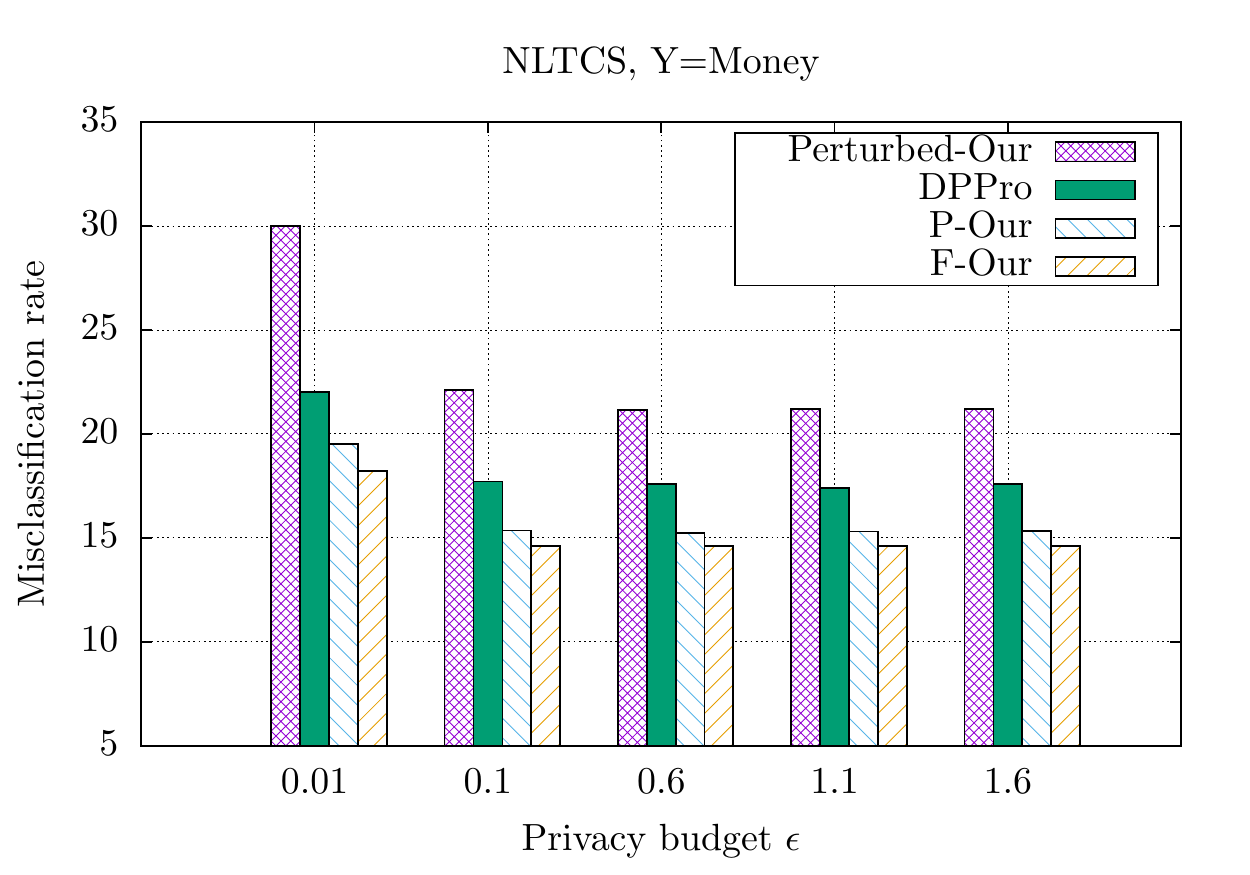}}
\end{minipage}
\begin{minipage}[t]{0.24\textwidth}
\centering
\subfigure[NLTCS, Y=Bathing]{
\includegraphics[width=\textwidth]{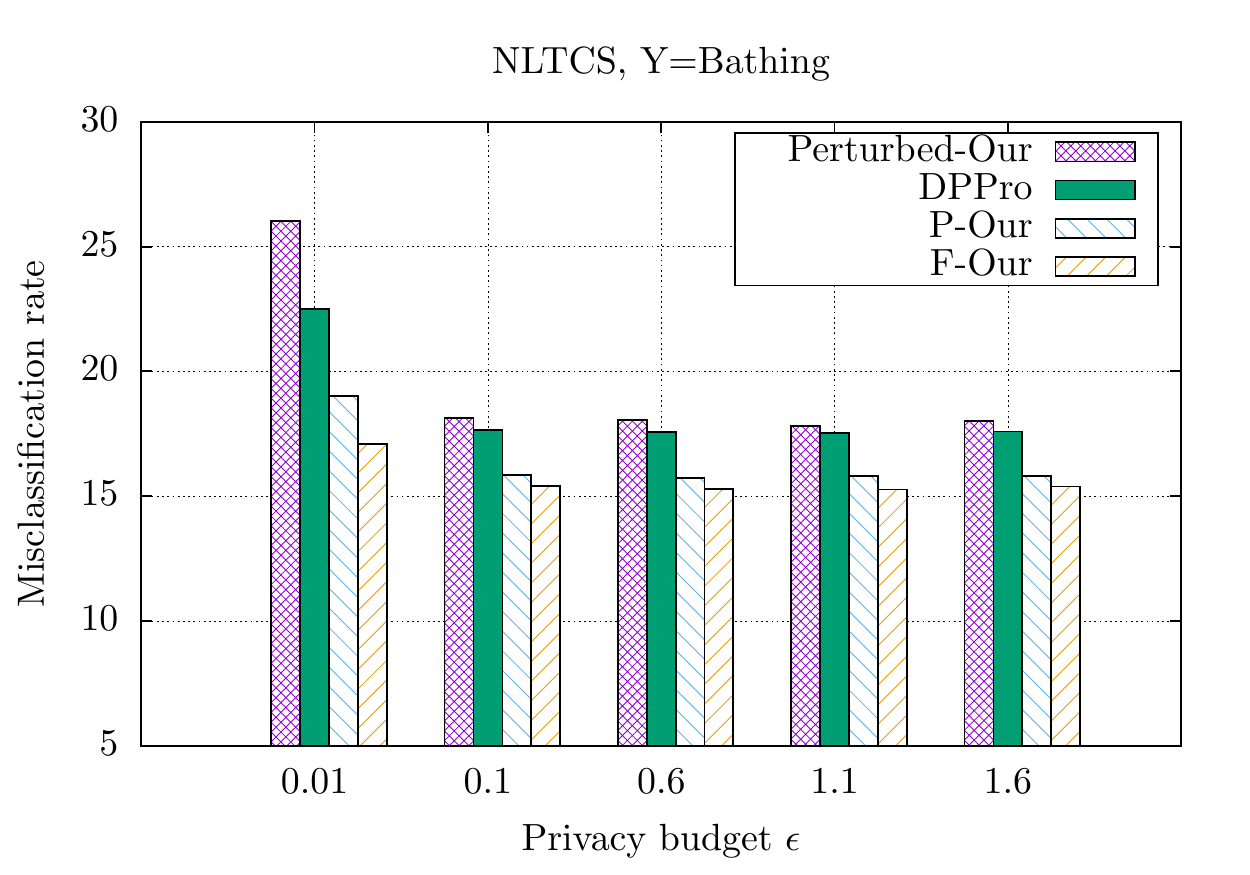}}
\end{minipage}
%\caption{Results of SVM with different $\epsilon$ on NLTCS.}
%\label{svm-nltcs}

\begin{minipage}[t]{0.24\textwidth}
\centering
\subfigure[BR2000, Y=Dwelling]{
\includegraphics[width=\textwidth]{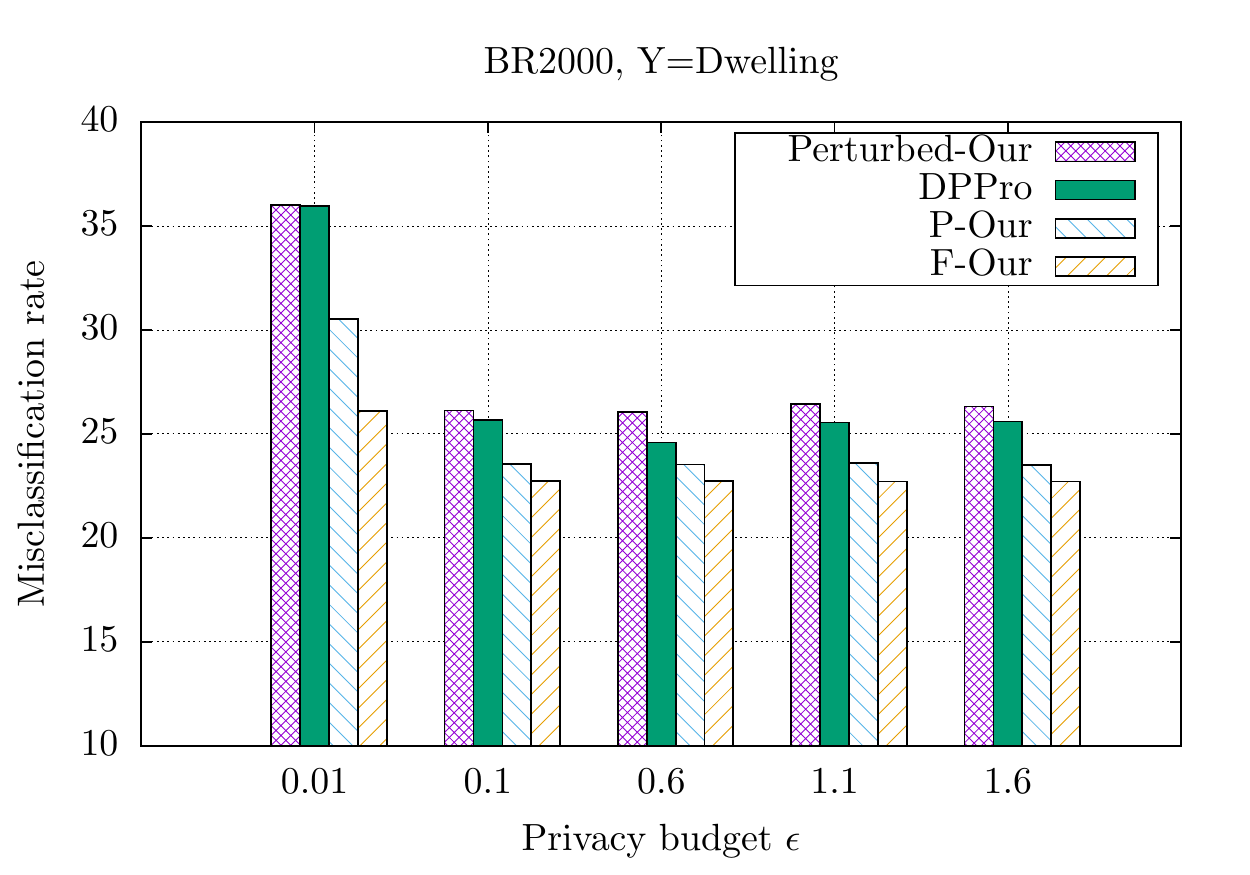}}
\end{minipage}
\begin{minipage}[t]{0.24\textwidth}
\centering
\subfigure[BR2000, Y=Religion]{
\includegraphics[width=\textwidth]{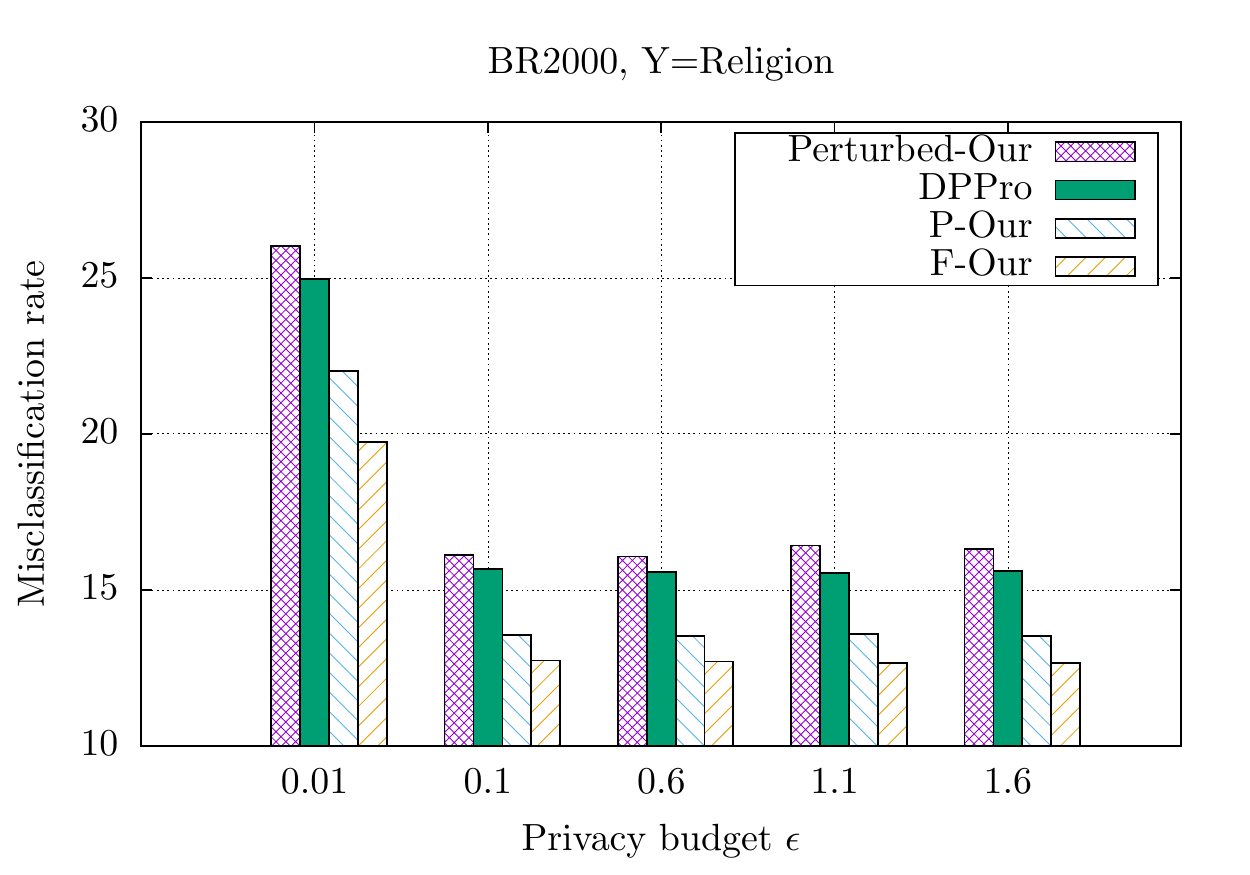}}
\end{minipage}
%\caption{Results of SVM with different $\epsilon$ on BR2000.}
%\label{svm-br2000}
%\end{figure}

%\begin{figure}[!htp]
\begin{minipage}[t]{0.24\textwidth}
\centering
\subfigure[Adult, Y=Gender]{
\includegraphics[width=\textwidth]{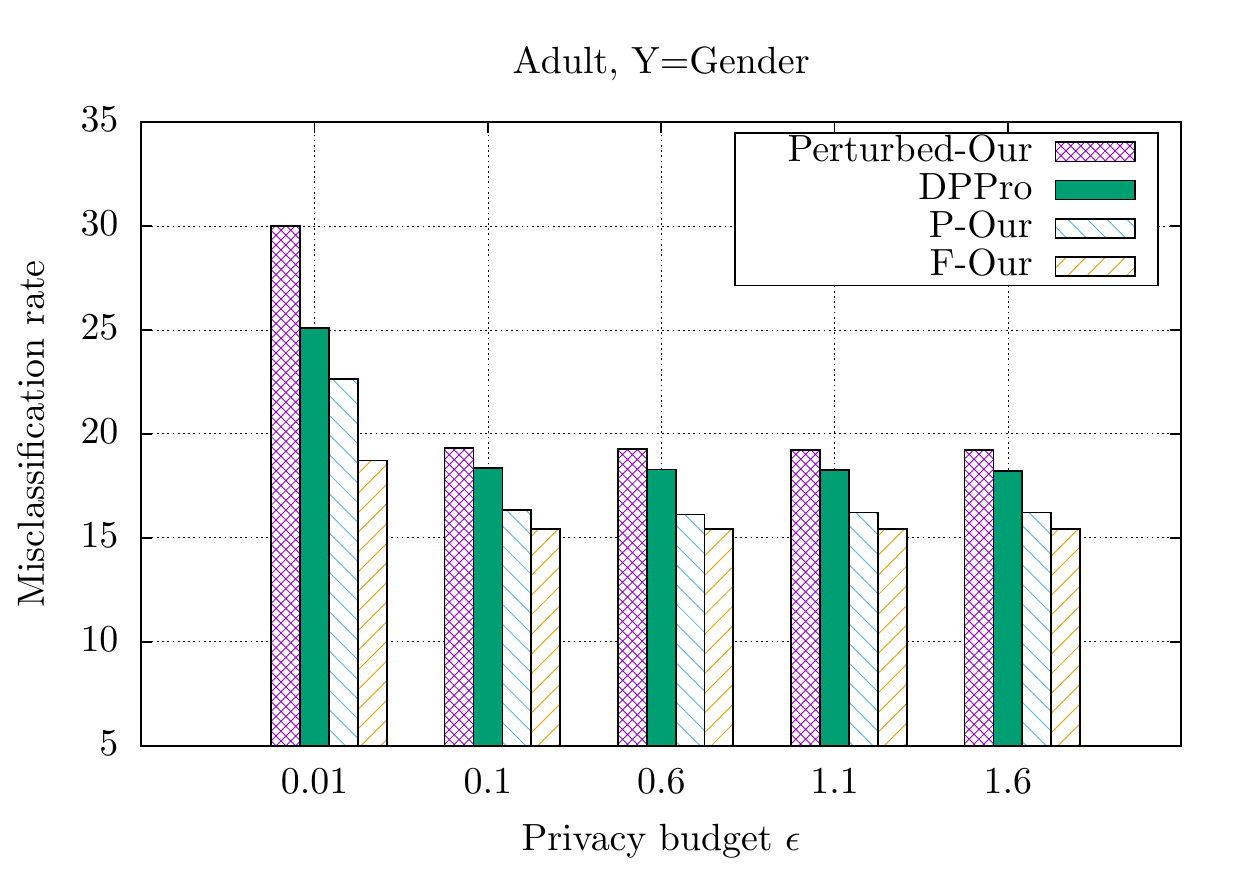}}
\end{minipage}
\begin{minipage}[t]{0.24\textwidth}
\centering
\subfigure[Adult, Y=Salary]{
\includegraphics[width=\textwidth]{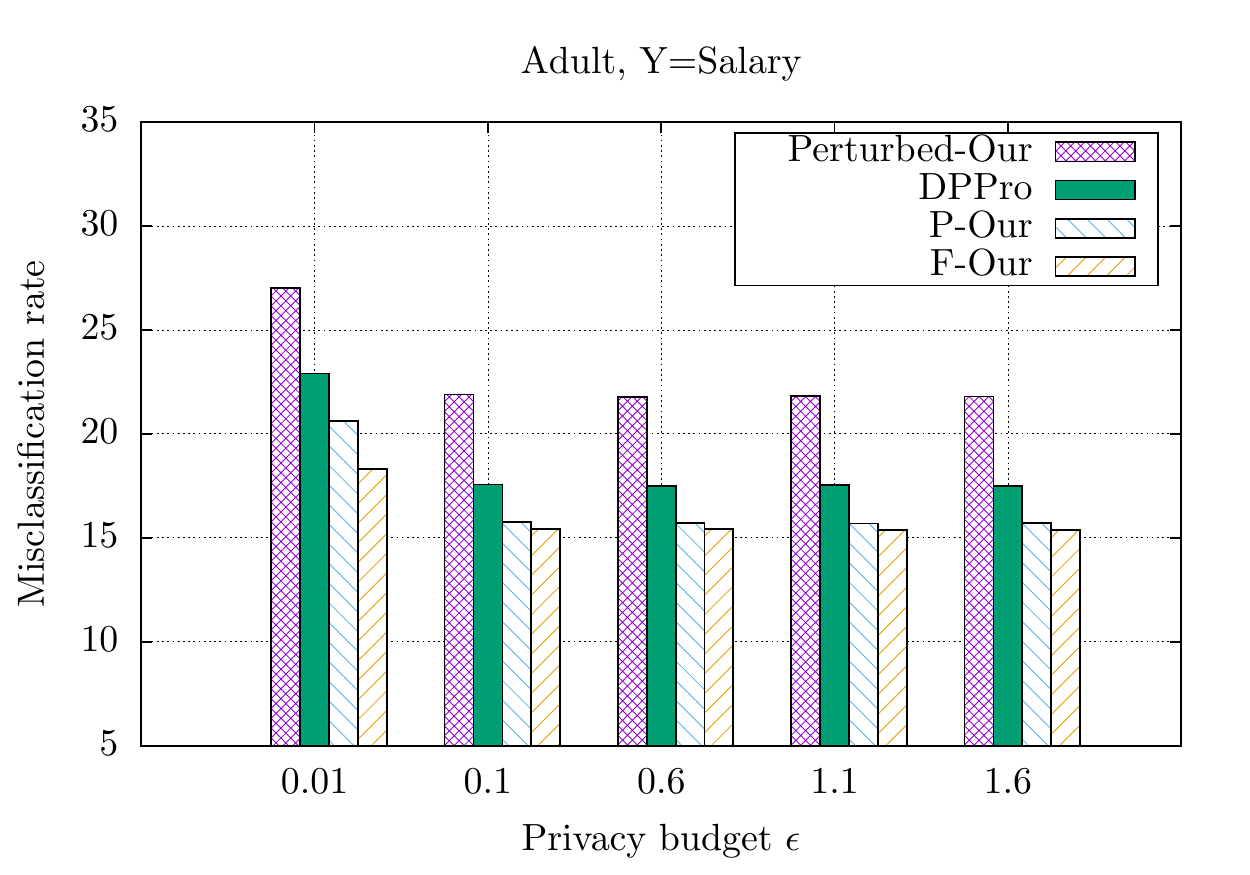}}
\end{minipage}
\caption{Results of SVM with different $\epsilon$}
\label{fig5}
\end{figure}

For the second task, we evaluate the performance of DPPro, and Perturbed data (Perturbed-Our), partially recovered data (P-Our) and fully recovered data (F-Our) of our ML-DPCS, for SVM classification. Figure \ref{fig5} shows the misclassification rate of each data copy under different privacy budgets. %The error of Non-Private remains unchanged for all $\epsilon$, since it denotes the SVM classification error without differential privacy. 
As shown in Figure \ref{fig5}, we can obtain the same conclusion with the results of $L_2$ error. Noting that the data utility of the perturbed data of ML-DPCS is lower than DPPro, while both P-Our and F-Our achieve better performance than DPPro on all datastets. Moreover, one can see that the misclassification rate decreases faster when $\epsilon$ increases from $0.01$ to $0.1$, and the decrease of the misclassification rate is not obvious when $\epsilon$ is larger than $0.6$. This indicates that a higher privacy level with a small $\epsilon$ leads to lower data utility.

\section{Conclusions and Future Research}\label{sec:con}

In this paper, we propose a differentially private data publication mechanism ML-DPCS achieving multi-level data utility with the help of compressive sampling theory. Specifically, the $n$-dimensional data can be projected into a lower $m$-dimensional space, then deliberately designed noise is added to provide differential privacy guarantee. Secondly, we selectively obfuscate the measurement vector by adding linearly encoded noise. Then we provide different data reconstruction algorithms for users with different authorization levels. Comprehensive experiments on four real-world datasets demonstrate that ML-DPCS achieves multi-level data utility, of which two levels of data utility outperform existing methods while providing differential privacy. In our future research, we intend to carefully design the adding noise to improve the data utility of perturbed data, and apply this framework on other privacy-preserving applications such as network/graph data.

\section*{Acknowledgment}

This work was partially supported by the US National Science Foundation under grant CNS-1704397. % and the National Science Foundation of China under grants 61832012, 61771289 and 61672321.

\bibliographystyle{IEEEtran}      %IEEEtran为给定模板格式定义文件名
\bibliography{ref}

\end{document}